\documentclass[11pt,letterpaper]{article}
\usepackage{times,mathptmx}
\DeclareMathAlphabet{\mathcal}{OMS}{cmsy}{m}{n}
\usepackage{fullpage}
\usepackage{amsmath,amsfonts,amssymb,amsthm,amscd}
\usepackage{hyperref,cite,url}
\newcommand{\svs}{\vspace{0.7mm}}
\newcommand{\vs}{\vspace{1.5mm}}
\newtheorem{theorem}{Theorem}[section]
\newtheorem{claim}[theorem]{Claim}

\newtheorem{lemma}[theorem]{Lemma}
\newtheorem{corollary}[theorem]{Corollary}

\newtheorem{remark}[theorem]{Remark}
\newtheorem{assumption}[theorem]{Assumption}
\newcommand{\G}{\mathbb{G}}
\newcommand{\Z}{\mathbb{Z}}
\newcommand{\bits}{\{0,1\}}
\newcommand{\Adv}{\textbf{Adv}}
\newcommand{\mc}[1]{\mathcal{#1}}
\newcommand{\tb}[1]{\textbf{#1}}

\title{Security Analysis of the Unrestricted Identity-Based\\ Aggregate
Signature Scheme}

\author{
    Kwangsu Lee\footnote{Korea University, Korea.
        Email: \texttt{guspin@korea.ac.kr}.}
    \and
    Dong~Hoon Lee\footnote{Korea University, Korea.
        Email: \texttt{donghlee@korea.ac.kr}.}
}

\date{}

\begin{document}

\maketitle

\begin{abstract}
Aggregate signatures allow anyone to combine different signatures signed by
different signers on different messages into a single short signature. An
ideal aggregate signature scheme is an identity-based aggregate signature
(IBAS) scheme that supports full aggregation since it can reduce the total
transmitted data by using an identity string as a public key and anyone can
freely aggregate different signatures. Constructing a secure IBAS scheme
that supports full aggregation in bilinear maps is an important open
problem. Recently, Yuan {\it et al.} proposed an IBAS scheme with full
aggregation in bilinear maps and claimed its security in the random oracle
model under the computational Diffie-Hellman assumption. In this paper, we
show that there exists an efficient forgery attacker on their IBAS scheme
and their security proof has a serious flaw.
\end{abstract}

\vs \noindent {\bf Keywords:} Identity-based signature, Aggregate signature,
Security analysis, Bilinear map.

\section{Introduction}

Aggregate signature schemes allow anyone to combine $n$ different signatures
on different $n$ messages signed by different $n$ signers into a single short
aggregate signature. The main advantage of aggregate signature schemes is to
reduce the communication and storage overhead of signatures by compressing
these signatures into a single signature. The application of aggregate
signature schemes includes secure routing protocols, public-key
infrastructure systems, and sensor networks. Boneh {\it et al.}
\cite{BonehGLS03} proposed the first full aggregate signature scheme in which
anyone can combine different signatures in bilinear groups and proved its
security in the random oracle model. After that, Lysyanskaya {\it et al.}
\cite{LysyanskayaMRS04} constructed a sequential aggregate signature scheme
such that a signature can be combined in sequential order, and Gentry and
Ramzan \cite{GentryR06} proposed a synchronized aggregate signature scheme
such that all signers should share synchronized information. There are many
other aggregate signature schemes with different properties \cite{LuOSSW06,
AhnGH10,BoldyrevaGOY10,Schroder11,GerbushLOW12,LeeLY13s,LeeLY13a,
HohenbergerSW13}.

Although aggregate signature schemes can reduce the size of signatures by
aggregation, they usually cannot reduce the total amount of transmitted data
significantly since a verifier should retrieve all public keys of the
signers. Therefore, reducing the size of public keys is also an important
issue in aggregate signature schemes \cite{Schroder11,LeeLY13s,LeeLY13a}. An
ideal solution for this problem is to use an identity-based aggregate
signature (IBAS) scheme since it uses an already known identity string as the
public key of a user \cite{GentryR06}. However, there is only one IBAS scheme
with full aggregation that was proposed by Hohenberger {\it et al.}
\cite{HohenbergerSW13} in multilinear maps. The multilinear map is an
attractive tool for cryptographic constructions, but it is currently
impractical since it's basis is a leveled homomorphic encryption scheme
\cite{GargGH13}. There are some IBAS schemes in bilinear maps, but these IBAS
schemes only support sequential aggregation or synchronized aggregation
\cite{GentryR06,BoldyrevaGOY10,GerbushLOW12}. Therefore, construction an IBAS
scheme with full aggregation in bilinear maps is an important open problem.

The main reason for the difficulty of devising an IBAS scheme with full
aggregation is that it seems not easy to find a way to aggregate the
randomness of all signers in which each randomness of a signer is used to
hide the private key of each signer in a signing process \cite{GentryR06}.
For this reason, current IBAS schemes only support synchronized aggregation
or sequential aggregation to aggregate the randomness of all signers
\cite{GentryR06,BoldyrevaGOY10}. Additionally, designing a secure IBAS scheme
is not a easy task since even the original version of Boldyreva {\it et
al.}'s IBAS scheme \cite{BoldyrevaGOY07} was broken by Hwang {\it et al.}
\cite{HwangLY09} and then it was corrected later. Recently, Yuan {\it et al.}
proposed an IBAS scheme with full aggregation in bilinear maps and claimed it
security in random oracle models \cite{YuanZH14}. The authors first proposed
an IBS scheme in bilinear maps and constructed an IBAS scheme from the IBS
scheme. To prove the security their IBS scheme, the authors claimed that the
security of their IBS scheme can be proven under the computational
Diffie-Hellman (CDH) assumption by using Forking Lemma in the random oracle
model.

In this paper, we show that the IBS and IBAS schemes of Yuan {\it et al.} are
not secure at all. First, we show that there exists a universal forgery
attack against the IBS scheme of Yuan {\it et al.} by using two signatures.
This forgery attack also applies to their IBAS scheme. One may wonder that
our forgery attack contradicts their claims of the security of the schemes.
To solve this, we next show that the security proof of Yuan {\it et al.}'s
IBS scheme has a serious flaw. The security proof of the IBS scheme
essentially use the fact that two signatures that are obtained by using
Forking Lemma have the same randomness in signatures. However, we show that
the forged signatures of an adversary cannot satisfy this condition since the
signature of Yuan {\it et al.}'s IBS scheme is publicly re-randomizable.

This paper is organized as follows: In Section 2, we review bilinear groups
and the IBS and IBAS schemes of Yuan {\it et al.} In Section 3, we present a
universal forgery against the IBS scheme. In Section 4, we analyze the
security proof of Yuan {\it et al.}'s IBS scheme and show that the proof has
a serious flaw.

\section{The IBAS Scheme of Yuan {\it et al.}}

In this section, we review bilinear groups and the IBS and IBAS schemes of
Yuan {\it et al.} \cite{YuanZH14}.

\subsection{Bilinear Groups and Complexity Assumptions}

Let $\G$ and $\G_T$ be two multiplicative cyclic groups of same prime order
$p$ and $g$ be a generator of $\G$. The bilinear map $e:\G \times \G
\rightarrow \G_{T}$ has the following properties:
\begin{enumerate}
\item Bilinearity: $\forall u, v \in \G$ and $\forall a,b \in \Z_p$,
    $e(u^a,v^b)=e(u,v)^{ab}$.
\item Non-degeneracy: $\exists g$ such that $e(g,g)$ has order $p$, that
    is, $e(g,g)$ is a generator of $\G_T$.
\end{enumerate}
We say that $\G$ is a bilinear group if the group operations in $\G$ and
$\G_T$ as well as the bilinear map $e$ are all efficiently computable.
Furthermore, we assume that the description of $\G$ and $\G_T$ includes
generators of $\G$ and $\G_T$ respectively.

\begin{assumption}[Computational Diffie-Hellman, CDH]
Let $(p, \G, \G_T, e)$ be a description of the bilinear group of prime order
$p$. Let $g$ be generators of subgroups $\G$. The CDH assumption is that if
the challenge tuple $D = \big( (p, \G, \G_T, e), g, g^a, g^b \big)$ is given,
no PPT algorithm $\mc{A}$ can output $g^{ab} \in \G$ with more than a
negligible advantage. The advantage of $\mc{A}$ is defined as
$\Adv_{\mc{A}}^{CDH} (\lambda) = \Pr[\mc{A}(D) = g^{ab}]$ where the
probability is taken over random choices of $a, b \in \Z_p$.
\end{assumption}

\subsection{The Identity-Based Signature Scheme}

The IBS scheme consists of \tb{Setup}, \tb{GenKey}, \tb{Sign}, and
\tb{Verify} algorithms. The IBS scheme of Yuan {\it et al.} \cite{YuanZH14}
is described as follows:

\begin{description}
\item [\tb{Setup}($1^{\lambda}$):] This algorithm takes as input a security
    parameter $1^{\lambda}$. It generates bilinear groups $\G, \G_T$ of
    prime order $p$. Let $g$ be a random generator of $\G$. It chooses
    random exponents $s_1, s_2 \in \Z_p^*$ and two cryptographic hash
    functions $H_1 : \bits^* \rightarrow \G$ and $H_2 : \bits^* \rightarrow
    \Z_p^*$. It outputs a master key $MK = (s_1, s_2)$ and public
    parameters $PP = \big( (p, \G, \G_T, e),~ g,~ g_1 = g^{s_1}, g_2 =
    g^{s_2},~ H_1, H_2 \big)$.

\item [\tb{GenKey}($ID, MK, PP$):] This algorithm takes as input an
    identity $ID \in \bits^*$, the master key $MK = (s_1, s_2)$, and the
    public parameters $PP$. It outputs a private key $SK_{ID} = \big( D_1 =
    H_1(ID)^{s_1} ,~ D_2 = H_1(ID)^{s_2} \big)$.

\item [\tb{Sign}($M, SK_{ID}, PP$):] This algorithm takes as input a
    message $M \in \bits^*$, a private key $SK_{ID} = (D_1, D_2)$, and the
    public parameters $PP$. It selects a random exponent $r \in \Z_p^*$ and
    computes $h = H_2(ID \| M)$. It outputs a signature $\sigma = \big( U =
    g^r,~ V = D_1^{h} \cdot g_1^{r},~ W = D_2 \cdot g_2^{r} \big)$.

\item [\tb{Verify}($\sigma, ID, M, PP$):] This algorithm takes as input a
    signature $\sigma = (U, V, W)$, an identity $ID \in \bits^*$, a message
    $M \bits^*$, and the public parameters $PP$. It computes $h = H(ID \|
    M)$ and checks whether $e(V, g) \stackrel{?}{=} e(H_1(ID)^{h} \cdot U,
    g_1)$ and $e(W, g) \stackrel{?}{=} e(H_1(ID) \cdot U, g_2)$. If both
    equations hold, then it outputs $1$. Otherwise, it outputs $0$.
\end{description}

\begin{claim}[\cite{YuanZH14}]
The above IBS scheme is existentially unforgeable under chosen message
attacks in the random oracle model if the CDH assumption holds.
\end{claim}

\begin{remark}
The original IBS and IBAS schemes of Yuan {\it et al.} is described in the
addictive notation in bilinear groups. However, in this paper, we use the
multiplicative notation instead of the addictive notation for the notational
simplicity.
\end{remark}

\begin{remark}
The signature of Yuan {\it et al.}'s IBS scheme is publicly re-randomizable.
If $\sigma = (U, V, W)$ is a valid signature, then a re-randomized signature
$\sigma' = (U \cdot g^{r'}, V \cdot g_1^{r'}, W \cdot g_2^{r'})$ is also a
valid one where $r'$ is a random exponent in $\Z_p^*$.
\end{remark}

\subsection{The Identity-Based Aggregate Signature Scheme}

The IBAS scheme consists of \tb{Setup}, \tb{GenKey}, \tb{Sign}, \tb{Verify},
\tb{Aggregate}, and \tb{AggVerify} algorithms. The \tb{Setup}, \tb{GenKey},
\tb{Sign}, and \tb{Verify} algorithms of Yuan {\it et al.}'s IBAS scheme is
the same as those of their IBS scheme. The IBAS scheme of Yuan {\it et al.}
\cite{YuanZH14} is described as follows:

\begin{description}
\item [\tb{Aggregate}($\sigma_1, \sigma_2, S_1, S_2, PP$):] This algorithm
    takes a input a signature $\sigma_1 = (U_1, V_1, W_1)$ on a multiset
    $S_1 = \{ (ID_{1,1}, M_{1,1}), \ldots, (ID_{1,n_1}, M_{1,n_1}) \}$ of
    identity and message pairs, a signature $\sigma_2 = (U_2, V_2, W_2)$ on
    a multiset $S_2 = \{ (ID_{2,1}, M_{2,1}), \ldots, (ID_{2,n_2},
    M_{2,n_2}) \}$ of identity and message pairs, and the public parameters
    $PP$. It outputs an aggregate signature $\sigma = \big( U = U_1 \cdot
    U_2, V = V_1 \cdot V_2, W = W_1 \cdot W_2 \big)$ on the multiset $S =
    S_1 \cup S_2$.

\item [\tb{AggVerify}($\sigma, S, PP$):] This algorithm takes as input an
    aggregate signature $\sigma = (U, V, W)$, a multiset $S = \{ (ID_1,
    M_1), \ldots, (ID_n, M_n) \}$ of identity and message pairs, and the
    public parameters $PP$. It computes $h_i = H(ID_i \| M_i)$ for $i=1,
    \ldots, n$ and checks whether $e(V, g) \stackrel{?}{=}
    e(\prod_{i=1}^{n} H_1(ID_i)^{h_i} \cdot U, g_1)$ and $e(W, g)
    \stackrel{?}{=} e(\prod_{i=1}^{n} H_1(ID_i) \cdot U, g_2)$. If both
    equations hold, then it outputs $1$. Otherwise, it outputs $0$.
\end{description}

\begin{claim}[\cite{YuanZH14}]
The above IBAS scheme is existentially unforgeable under chosen message
attacks in the random oracle model if the underlying IBS scheme is
unforgeable under chosen message attacks.
\end{claim}

\section{Forgery Attacks on the IBAS Scheme}

In this section, we show that the IBS and IBAS schemes of Yuan {\it et al.}
are not secure at all by presenting an efficient forgery algorithm. In fact,
our forgery algorithm is universal since anyone who has two valid signatures
on the same identity with different messages can generate a forge signature
on the same identity with any message of its choice.

\begin{lemma}
There exists a probabilistic polynomial-time (PPT) algorithm $\mc{F}$ that
can forge the IBS scheme of Yuan {\it et al.} except negligible probability
if $\mc{F}$ makes just two signature queries.
\end{lemma}

\begin{proof}
The basic idea of our forgery attack is that if a forger obtains two valid
signature on an identity, then a linear combination of these signatures can
be another valid signature by carefully choosing scalar values. A forgery
algorithm $\mc{F}$ is described as follows:

\begin{enumerate}
\item $\mc{F}$ randomly selects a target identity $ID^*$ and two different
    messages $M_1$ and $M_2$. It obtains a signature $\sigma_1 = (U_1, V_1,
    W_1)$ on the pair $(ID^*, M_1)$ and a signature $\sigma_2 = (U_2, V_2,
    W_2)$ on the pair $(ID^*, M_2)$ from the signature oracle.

\item It randomly selects a target message $M^*$ for a forged signature.
    Next, it computes $h_1 = H_2(ID^* \| M_1)$, $h_2 = H_2(ID^* \| M_2)$,
    and $h^* = H_2(ID^* \| M^*)$. It computes two exponents $\delta_1,
    \delta_2$ that satisfy the following equation
    \begin{align*}
    \left[ \begin{array}{cc}
    h_1 & h_2 \\    1   & 1   \\
    \end{array} \right]
    \left[ \begin{array}{c}
    \delta_1 \\     \delta_2 \\
    \end{array} \right]
    =
    \left[ \begin{array}{c}
    h^* \\          1 \\
    \end{array} \right]
    \mod p.
    \end{align*}
    Note that if $h_1 \neq h_2$, then $\delta_1, \delta_3$ can be computed
    by using Linear Algebra since the determinant $h_1 - h_2$ of the left
    matrix is not zero.

\item Finally, $\mc{F}$ outputs a forged signature $\sigma^*$ on the
    identity and message pair $(ID^*, M^*)$ as
    \begin{align*}
    \sigma^* = \big(
        U^* = U_1^{\delta_1} \cdot U_2^{\delta_2},~
        V^* = V_1^{\delta_1} \cdot V_2^{\delta_2},~
        W^* = W_1^{\delta_1} \cdot W_2^{\delta_2}
    \big).
    \end{align*}
\end{enumerate}

To finish the proof, we should show that the forger $\mc{F}$ outputs a
(forged) signature with non-negligible probability and the forged signature
passes the verification algorithm. We known that $\mc{F}$ always outputs a
signature if $h_1 \neq h_2$. Because $H_2$ is a collision-resistant hash
function and $M_1 \neq M_2$, we have that $h_1 \neq h_2$ except negligible
probability. Now we should show that the forged signature is correct by the
verification algorithm. Let $r_1, r_2$ be the randomness of $\sigma_1,
\sigma_2$ respectively. The correctness of the forged signature is easily
verified as follows
    \begin{align*}
    U^* &= U_1^{\delta_1} \cdot U_2^{\delta_2}
         = g^{r_1 \delta_1 + r_2 \delta_2}
         = g^{r^*}, \\
    V^* &= V_1^{\delta_1} \cdot V_2^{\delta_2}
         = (H_1(ID^*)^{s_1 h_1} g_1^{r_1})^{\delta_1} \cdot
           (H_1(ID^*)^{s_1 h_2} g_1^{r_2})^{\delta_2} \\
        &= H_1(ID^*)^{s_1 (h_1 \delta_1 + h_2 \delta_2)} g_1^{r_1 \delta_1 + r_2 \delta_2}
         = H_1(ID^*)^{s_1 h^*} g_1^{r^*}, \\
    W^* &= W_1^{\delta_1} \cdot W_2^{\delta_2}
         = (H_1(ID^*)^{s_2} g_2^{r_1})^{\delta_1} \cdot
           (H_1(ID^*)^{s_2} g_2^{r_2})^{\delta_2} \\
        &= H_1(ID^*)^{s_2 (\delta_1 + \delta_2)} g_2^{r_1 \delta_1 + r_2 \delta_2}
         = H_1(ID^*)^{s_2} g_2^{r^*}
    \end{align*}
where the randomness of the forged signature is defined as $r^* = r_1
\delta_1 + r_2 \delta_2 \mod p$. This completes the proof.
\end{proof}

\begin{corollary}
There exists a PPT algorithm $\mc{F}$ that can forge the IBAS scheme of Yuan
{\it et al.} except negligible probability if $\mc{F}$ makes just two
signature queries.
\end{corollary}

The proof of this corollary is trivial from the proof of the previous Lemma
since the IBAS scheme uses the IBS scheme as the underlying signature scheme.
We omit the proof.

\section{Our Analysis of the Security Proof}

From the forgery attack in the previous section, it is evident that the IBS
and IBAS schemes of Yuan {\it et al.} are not secure. However, Yuan {\it et
al.} claimed that their IBS scheme is secure in the random oracle model under
the CDH assumption by using Forking Lemma in \cite{YuanZH14}. In this
section, we analyze the security proof of Yuan {\it et al.} and show that
there is a critical flaw in their security proof that uses Forking Lemma.

\subsection{The Original Proof}

In this subsection, we briefly review the security proof of Yuan {\it et
al.}'s IBS scheme \cite{YuanZH14} that solves the CDH problem by using
Forking Lemma \cite{PointchevalS00,BellareN06}.

Suppose there exists an adversary $\mc{A}$ that outputs a forged signature
for the IBS scheme with a non-negligible advantage. A simulator $\mc{B}$ that
solves the CDH problem using $\mc{A}$ is given: a challenge tuple $D = ((p,
\G, \G_T, e), g, g^a, g^b)$. Then $\mc{B}$ that interacts with $\mc{A}$ is
described as follows:

\noindent \tb{Setup}: $\mc{B}$ chooses a random exponent $s_2 \in \Z_p^*$ and
maintains $H_1$-list and $H_2$-list for random oracles. It implicitly sets
$s_1 = a$ and publishes the public parameters $PP = \big( (p, \G, \G_T, e),
g,~ g_1 = g^a, g_2 = g^{s_2},~ H_1, H_2 \big)$.

\svs \noindent \tb{Hash Query}: If this is an $H_1$ hash query on an identity
$ID_i$, then $\mc{B}$ handles this query as follows: If the identity $ID_i$
already appears in $H_1$-list, then it responds with the value in the list.
Otherwise, it picks a random coin $c \in \bits$ with $\Pr[c=0] = \delta$ for
some $\delta$ and proceeds as follows: If $c = 0$, then it chooses $t_i \in
\Z_p^*$ and sets $Q_i = (g^b)^{t_i}$. If $c = 1$, then it chooses $t_i \in
\Z_p^*$ and sets $Q_i = g^{t_i}$. Next, it adds $(ID_i, t_i, c, Q_i)$ to
$H_1$-list and responds to $\mc{A}$ with $H_1(ID_i) = Q_i$.

If this is an $H_2$ hash query on an identity $ID_i$ and a message $M_i$,
then $\mc{B}$ handles this query as follows: If the tuple $(ID_i, M_i)$
already appears on $H_2$-list, then it responds with the value in the list.
Otherwise, it randomly chooses $h_i \in \Z_p^*$, add $(ID_i, M_i, h_i)$ to
$H_2$-list, and responds with $H_2(ID_i \| M_i) = h_i$.

\svs \noindent \tb{Private-Key Query}: $\mc{B}$ handles a private key query
for an identity $ID_i$ as follows:
It first retrieves $(ID_i, t_i, c, Q_i)$ from the $H_1$-list. If $c = 0$,
then it aborts the simulation since it cannot create a private key.
Otherwise, it creates a private key $SK_{ID_i} = (D_1 = g_1^{t_i}, D_2 =
g_2^{t_i})$ and responds to $\mc{A}$ with $SK_{ID_i}$.

\svs \noindent \tb{Signature Query}: $\mc{B}$ handles a signature query on an
identity $ID_i$ and a message $M_i$ as follows: It randomly chooses $r' \in
\Z_p^*$ and computes $h = H_2(ID_i \| M_i)$. Next, it responds to $\mc{A}$
with a signature $\sigma = \big( U = g^{r'} \cdot H_1(ID_i)^{-h},~ V =
g_1^{r'},~ W = H_1(ID_i)^{s_2} \cdot U^{s_2} \big)$.

\svs \noindent \tb{Output}: $\mc{A}$ finally outputs a forged signature
$\sigma^* = (U^*, V^*, W^*)$ on an identity $ID^*$ and a message $M^*$.

To solve the CDH problem, $\mc{B}$ retrieves the tuple $(ID^*, t^*, c^*,
Q^*)$ from the $H_1$-list. If $c^* \neq 0$, then it aborts since it cannot
extract the CDH value. Otherwise, it obtains two valid signatures $\sigma_1^*
= (U_1^*, V_1^*, W_1^*)$ and $\sigma_2^* = (U_2^*, V_2^*, W_2^*)$ on the same
identity and message tuple $(ID^*, M^*)$ such that $U_1^* = U_2^*$ and $h_1^*
\neq h_2^*$ by applying Forking Lemma. That is, it replays $\mc{F}$ with the
same random tape but different choice of the random oracle $H_2$. If $U_1^* =
U_2^*$, then we have the following equation
    \begin{align*}
    V_1^* \cdot (V_2^*)^{-1}
    &= H_1(ID^*)^{s_1 h_1^*} g_1^{s_1 r^*} \cdot
       (H_1(ID^*)^{s_1 h_2^*} g_1^{s_1 r^*})^{-1} \\
    &= H_1(ID^*)^{s_1 (h_1^* - h_2^*)}
     = (g^{ab})^{t^* (h_1^* - h_2^*)}.
    \end{align*}
Thus, $\mc{B}$ can compute the CDH value as $(V_1^* \cdot
(V_2^*)^{-1})^{1/(t^* (h_1^* - h_2^*))}$ if $h_1^* \neq h_2^* \mod p$.

\subsection{A Non-Extractable Forgery}

To extract the CDH value from forged signatures by applying Forking Lemma, it
is essential for the simulator to obtains two valid signatures $\sigma_1^*$
and $\sigma_2^*$ such that $U_1^* = U_2^*$ and $h_1^* \neq h_2^*$. By
replaying a forgery with the same random tape with different choice of random
oracle $H_2$, it is possible for a simulator to obtain two valid signatures
$\sigma_1^* = (U_1^*, V_1^*, W_1^*)$ and $\sigma_2^* = (U_2^*, V_2^*, W_2^*)$
with $h_1^* \neq h_2^*$ because of Forking Lemma. However, we show that the
probability of $U_1^* = U_2^*$ is negligible for some clever forgery.

\begin{lemma}
If there is a PPT algorithm $\mc{A}$ that can forge the IBS scheme of Yuan et
al., then there is another PPT algorithm $\mc{F}$ that can forge the IBS
scheme with almost the same probability except that the simulator of Yuan et
al. cannot extract the CDH value from the forged signatures of $\mc{F}$.
\end{lemma}

\begin{proof}
The basic idea of this proof is that anyone can re-randomize the signature of
Yuan {\it et al.}'s IBS scheme by using the public parameters. In this case,
even though a simulator use the same random tape for Forking Lemma, a forgery
output a forged signature $\sigma^*$ on an identity $ID^*$ and a message
$M^*$ after re-randomizing it by using the information $h^* = H_2(ID^* \|
M^*)$. Let $H':\bits^* \rightarrow \Z_p$ be a collision resistant hash
function that is not modeled as the random oracle. A new forgery $\mc{F}$
that uses $\mc{A}$ as a sub-routine is described as follows:
\begin{enumerate}
\item $\mc{F}$ is first given $PP$ and runs $\mc{A}$ by giving $PP$.
    $\mc{F}$ also handles the private key and signature queries of $\mc{A}$
    by using his own private key and signature oracles.

\item $\mc{A}$ finally outputs a forged signature $\sigma' = (U', V', W')$
    on an identity $ID^*$ and a message $M^*$.

\item $\mc{F}$ computes $h^* = H_2(ID^* \| M^*)$ and $h' = H'(U' \| h^*)$,
    and then it re-randomizes the forged signature as
    \begin{align*}
    \sigma^* = \big(
        U^* = U' \cdot g^{h'},~
        V^* = V' \cdot g_1^{h'},~
        W^* = W' \cdot g_2^{h'}
    \big).
    \end{align*}

\item Finally, $\mc{F}$ outputs $\sigma^* = (U^*, V^*, W^*)$ as the forged
    signature on an identity $ID^*$ and $M^*$.
\end{enumerate}

To finish the proof, we should show that the forged signature of $\mc{F}$ is
correct and the simulator of Yuan {\it et al.} cannot extract the CDH value
from the forged signatures by using Forking Lemma. Let $r'$ be the randomness
of $\sigma'$. The correctness of the forged signature is easily checked as
follows
    \begin{align*}
    U^* &= U' \cdot g^{h'}
         = g^{r' + h'}
         = g^{r^*}, \\
    V^* &= V' \cdot g_1^{h'}
         = H_1(ID^*)^{s_1 h^*} g_1^{r' + h'}
         = H_1(ID^*)^{s_1 h^*} g_1^{r^*}, \\
    W^* &= W' \cdot g_2^{h'}
         = H_1(ID^*)^{s_2} g_2^{r' + h'}
         = H_1(ID^*)^{s_2} g_2^{r^*}
    \end{align*}
where $h^* = H_2(ID^* \| M^*)$ and $r^* = r' + h'$.
To extract the CDH value from the forged signature of $\mc{F}$ by using
Forking Lemma, the simulator of Yuan {\it et al.} should obtain two valid
signatures $\sigma_1^* = (U_1^*, V_1^*, W_1^*)$ and $\sigma_2^* = (U_2^*,
V_2^*, W_2^*)$ on the same identity and message pair $(ID^*, M^*)$ such that
$U_1^* = U_2^*$ and $h_1^* \neq h_2^*$ after replaying $\mc{F}$ with the same
random tape but different choices of the hash oracle $H_2$. Let $\sigma_1^* =
(U_1^*, V_1^*, W_1^*)$ and $\sigma_2^* = (U_2^*, V_2^*, W_2^*)$ be the two
valid signatures obtained from $\mc{F}$ by using Forking Lemma and $\sigma'_1
= (U'_1, V'_1, W'_1)$ and $\sigma'_2 = (U'_2, V'_2, W'_2)$ be the original
signatures before the re-randomization of $\mc{F}$. If $h_1^* \neq h_2^*$,
then $h'_1 \neq h'_2$ except negligible probability since $H'$ is a
collision-resistance hash function and the inputs of this hash function are
different. From $h'_1 \neq h'_2$, we have $U_1^* \neq U_2^*$ except
negligible probability since $U'_1$ and $U'_2$ are re-randomized with
difference values $g^{h'_1}$ and $g^{h'_2}$ respectively. Therefore, the
event that the simulator obtains two valid signatures such that $U_1^* =
U_2^*$ and $h_1^* \neq h_2^*$ by using Forking Lemma only occurs with
negligible probability. This completes our proof.
\end{proof}

\subsection{Discussions}

From the above analysis, we know that the original IBS scheme of Yuan {\it et
al.} cannot be proven secure under the CDH assumption by applying Forking
Lemma since the signature is publicly re-randomizable. To fix this problem,
we may modify the IBS scheme to compute $h = H_2(U \| ID \| M)$ instead of $h
= H_2(ID \| M)$ where $U$ is the first element of a signature. In this case,
the signature of the modified IBS scheme is not re-randomizable since $U$ is
given to the input of $H_2$. Note that our forgery attack in the previous
section also does not work in this modified IBS scheme. However, this
modified IBS scheme does not lead to an IBAS scheme since each $U$ in
individual signatures cannot be aggregated. Note that if each $U$ is
aggregated, then a verifier cannot check the validity of an aggregate
signature since each $U$ is not given in the aggregate signature. Therefore,
there is no easy fix to solve the problem.

\section{Conclusion}

In this paper, we showed that the IBS and IBAS schemes of Yuan {\it et al.}
are not secure at all. We first presented an efficient forgery attack on the
IBS scheme and their security proof of the IBS scheme has a serious flaw. The
IBAS scheme is also not secure since the security of their IBAS scheme is
based on the security of their IBS scheme. Therefore, constructing an IBAS
scheme with full aggregation in bilinear maps is still left as an important
open problem.



\bibliographystyle{plain}
\bibliography{yzh-ibas-analysis}

\begin{thebibliography}{10}

\bibitem{AhnGH10}
Jae~Hyun Ahn, Matthew Green, and Susan Hohenberger.
\newblock Synchronized aggregate signatures: new definitions, constructions and
  applications.
\newblock In {\em ACM Conference on Computer and Communications Security},
  pages 473--484, 2010.

\bibitem{BellareN06}
Mihir Bellare and Gregory Neven.
\newblock Multi-signatures in the plain public-key model and a general forking
  lemma.
\newblock In Ari Juels, Rebecca~N. Wright, and Sabrina De~Capitani
  di~Vimercati, editors, {\em CCS 2006}, pages 390--399. {ACM}, 2006.

\bibitem{BoldyrevaGOY07}
Alexandra Boldyreva, Craig Gentry, Adam O'Neill, and Dae~Hyun Yum.
\newblock Ordered multisignatures and identity-based sequential aggregate
  signatures, with applications to secure routing.
\newblock In Peng Ning, Sabrina De~Capitani di~Vimercati, and Paul~F. Syverson,
  editors, {\em ACM Conference on Computer and Communications Security}, pages
  276--285. ACM, 2007.

\bibitem{BoldyrevaGOY10}
Alexandra Boldyreva, Craig Gentry, Adam O'Neill, and Dae~Hyun Yum.
\newblock Ordered multisignatures and identity-based sequential aggregate
  signatures, with applications to secure routing.
\newblock Cryptology ePrint Archive, Report 2007/438, 2010.
\newblock \url{http://eprint.iacr.org/2007/438}.

\bibitem{BonehGLS03}
Dan Boneh, Craig Gentry, Ben Lynn, and Hovav Shacham.
\newblock Aggregate and verifiably encrypted signatures from bilinear maps.
\newblock In Eli Biham, editor, {\em EUROCRYPT 2003}, volume 2656 of {\em
  Lecture Notes in Computer Science}, pages 416--432. Springer, 2003.

\bibitem{GargGH13}
Sanjam Garg, Craig Gentry, and Shai Halevi.
\newblock Candidate multilinear maps from ideal lattices.
\newblock In Thomas Johansson and Phong~Q. Nguyen, editors, {\em EUROCRYPT
  2013}, volume 7881 of {\em Lecture Notes in Computer Science}, pages 1--17.
  Springer, 2013.

\bibitem{GentryR06}
Craig Gentry and Zulfikar Ramzan.
\newblock Identity-based aggregate signatures.
\newblock In Moti Yung, Yevgeniy Dodis, Aggelos Kiayias, and Tal Malkin,
  editors, {\em PKC 2006}, volume 3958 of {\em Lecture Notes in Computer
  Science}, pages 257--273. Springer, 2006.

\bibitem{GerbushLOW12}
Michael Gerbush, Allison~B. Lewko, Adam O'Neill, and Brent Waters.
\newblock Dual form signatures: An approach for proving security from static
  assumptions.
\newblock In Xiaoyun Wang and Kazue Sako, editors, {\em ASIACRYPT 2012}, volume
  7658 of {\em Lecture Notes in Computer Science}, pages 25--42. Springer,
  2012.

\bibitem{HohenbergerSW13}
Susan Hohenberger, Amit Sahai, and Brent Waters.
\newblock Full domain hash from (leveled) multilinear maps and identity-based
  aggregate signatures.
\newblock In Ran Canetti and Juan~A. Garay, editors, {\em CRYPTO 2013}, volume
  8042 of {\em Lecture Notes in Computer Science}, pages 494--512. Springer,
  2013.

\bibitem{HwangLY09}
Jung~Yeon Hwang, Dong~Hoon Lee, and Moti Yung.
\newblock Universal forgery of the identity-based sequential aggregate
  signature scheme.
\newblock In Wanqing Li, Willy Susilo, Udaya~Kiran Tupakula, Reihaneh
  Safavi-Naini, and Vijay Varadharajan, editors, {\em ASIACCS 2009}, pages
  157--160. ACM, 2009.

\bibitem{LeeLY13a}
Kwangsu Lee, Dong~Hoon Lee, and Moti Yung.
\newblock Aggregating cl-signatures revisited: Extended functionality and
  better efficiency.
\newblock In Ahmad-Reza Sadeghi, editor, {\em FC 2013}, volume 7859 of {\em
  Lecture Notes in Computer Science}, pages 171--188. Springer, 2013.

\bibitem{LeeLY13s}
Kwangsu Lee, Dong~Hoon Lee, and Moti Yung.
\newblock Sequential aggregate signatures with short public keys: Design,
  analysis and implementation studies.
\newblock In Kaoru Kurosawa and Goichiro Hanaoka, editors, {\em PKC 2013},
  volume 7778 of {\em Lecture Notes in Computer Science}, pages 423--442.
  Springer, 2013.

\bibitem{LuOSSW06}
Steve Lu, Rafail Ostrovsky, Amit Sahai, Hovav Shacham, and Brent Waters.
\newblock Sequential aggregate signatures and multisignatures without random
  oracles.
\newblock In Serge Vaudenay, editor, {\em EUROCRYPT 2006}, volume 4004 of {\em
  Lecture Notes in Computer Science}, pages 465--485. Springer, 2006.

\bibitem{LysyanskayaMRS04}
Anna Lysyanskaya, Silvio Micali, Leonid Reyzin, and Hovav Shacham.
\newblock Sequential aggregate signatures from trapdoor permutations.
\newblock In Christian Cachin and Jan Camenisch, editors, {\em EUROCRYPT 2004},
  volume 3027 of {\em Lecture Notes in Computer Science}, pages 74--90.
  Springer, 2004.

\bibitem{PointchevalS00}
David Pointcheval and Jacques Stern.
\newblock Security arguments for digital signatures and blind signatures.
\newblock {\em J. Cryptology}, 13(3):361--396, 2000.

\bibitem{Schroder11}
Dominique Schr{\"o}der.
\newblock How to aggregate the cl signature scheme.
\newblock In Vijay Atluri and Claudia D\'{\i}az, editors, {\em ESORICS 2011},
  volume 6879 of {\em Lecture Notes in Computer Science}, pages 298--314.
  Springer, 2011.

\bibitem{YuanZH14}
Yumin Yuan, Qian Zhan, and Hua Huang.
\newblock Efficient unrestricted identity-based aggregate signature scheme.
\newblock {\em PLoS ONE}, 9(10):e110100, 2014.

\end{thebibliography}

\end{document}